\documentclass[a4paper,leqno]{amsart}

\usepackage{latexsym}
\usepackage{amsmath}
\usepackage{mathrsfs}
\usepackage{amsthm}
\usepackage{amsfonts}
\usepackage{amssymb}
\usepackage{amscd}
\usepackage{bbm}
\usepackage{graphicx}
\usepackage{graphics}
\usepackage{latexsym}
\usepackage{color}

\newcommand{\C}{\mathbb C}

\newcommand{\N}{\mathbb N}
\newcommand{\R}{\mathbb R}
\newcommand{\eps}{\varepsilon}
\DeclareMathOperator{\st}{st}
\DeclareMathOperator{\supp}{supp}
\DeclareMathOperator{\RE}{Re}

\theoremstyle{plain}
\newtheorem{theorem}{Theorem}[section]
\newtheorem{lemma}[theorem]{Lemma}
\newtheorem{corollary}[theorem]{Corollary}

\theoremstyle{definition}
\newtheorem{definition}[theorem]{Definition}

\newtheorem{remark}[theorem]{Remark}
\newtheorem*{remark*}{Remark}

\newtheorem*{definition*}{Definition}
\newtheorem*{fact*}{Fact}
\newtheorem*{transfer*}{Transfer Principle}

\numberwithin{equation}{section}

\begin{document}

\title[Singular Schr\"odinger operators on the half-line]{A characterization of singular Schr\"odinger operators on the half-line}

\author{Raffaele Scandone}
\email{raffaele.scandone@gssi.com}
\address{Gran Sasso Science Institute, Via Crispi 7, 67100 L'Aquila, Italy}
\author{Lorenzo Luperi Baglini}
\email{lorenzo.luperi@unimi.it}
\address{Dipartimento di Matematica, Universit\`{a} di Milano, Via Saldini 50, 20133 Milano, Italy, supported
by grant P 30821-N35 of the Austrian Science Fund FWF.}
\author{Kyrylo Simonov}
\email{kyrylo.simonov@univie.ac.at}
\address{Fakult\"{a}t f\"{u}r Mathematik, Universit\"{a}t Wien, Oskar-Morgenstern-Platz 1, 1090 Vienna, Austria, supported
by grant P 30821-N35 of the Austrian Science Fund FWF.}

\begin{abstract}
We study a class of delta-like perturbations of the Laplacian on the half-line, characterized by Robin boundary conditions at the origin. Using the formalism of nonstandard analysis, we derive a simple connection with a suitable family of Schr\"odinger operators with potentials of very large (infinite) magnitude and very short (infinitesimal) range. As a consequence, we also derive a similar result for point interactions in the Euclidean space $\R^3$, in the case of radial potentials. Moreover, we discuss explicitly our results in the case of potentials that are linear in a neighbourhood of the origin.
\end{abstract}

\date{\today}

\maketitle


\section{Introduction}\label{intro}
In this paper, we study the behavior of Schr\"odinger operators with very short range potentials on the half-line, of the form
\begin{equation}\label{formal_operator}
-\Delta_D+\lambda V\Big(\frac{x}{\eps}\Big),
\end{equation}
where $-\Delta_D$ denotes the Dirichlet Laplacian on $L^2(\R^+)$, $V$ is a real-valued, bounded, compactly supported function, $\eps\ll 1$, and $\lambda:=\lambda(\eps)\in\R$. In particular, we are interested in the case when the potential is virtually zero-range, i.e.~when $\eps\to 0$, and when $\lambda$ becomes very large, i.e.~$\lambda\to\infty$ as $\eps\to 0$, in order to create a delta-like profile. Such models have been widely used in nuclear, atomic, and solid state physics \cite{Bethe_Peierls-1935,Bethe_Peierls-1935-np,Kronig-Pennig-1931,Thomas}, as they provide heuristic Hamiltonian for a non-relativistic quantum particle moving under the influence of a fixed impurity.

Singular problems can be fruitfully studied by means of non-Archimedean methods, which provide rigorous ways to handle infinitesimal and infinite quantities; in particular, in our problem they give the possibility to formalize the idea of letting the potential be zero ranged and delta-like by taking $\eps$ infinitesimal and $\lambda$ infinite in the expression \eqref{formal_operator}. In this paper, we will focus mostly on non-Archimedean methods based on nonstandard analysis, on similar lines of those contained in the seminal paper \cite{Albeverio-Fenstad-HoeghKrohn-1979_singPert_NonstAnal} (see also the recent paper \cite{Benci-Baglini-Simonov}), but there have also been other non-Archimedean approaches to the study of equations like (\ref{formal_operator}), for example, those based on different versions of Colombeau algebras (see e.g.~\cite{Dug, Gunt} and references therein). 

The classical construction of Schr\"odinger operators with a zero-range interaction (also known as point interaction) is based on a restriction-extension approach (see e.g.~\cite{GTV} and references therein): one first considers the closed, symmetric operator $T:=-\Delta|_{\mathcal{C}^{\infty}_c(\R^+)}$ on $L^2(\R^+)$, then the Krein-von Neumann theory \cite{Reed-Simon-II} guarantees that $T$ admits a one-parameter family of self-adjoint extensions $\{-\Delta_{\alpha}\}_{\alpha\in\R\cup\{\infty\}}$, where $-\Delta_{\alpha}$ acts as the free Laplacian, and functions $\psi\in\mathcal{D}(-\Delta_{\alpha})$ are qualified by the boundary condition $\alpha\psi(0)=\psi'(0)$ (the case $\alpha=+\infty$ corresponds to the Dirichlet Laplacian). More precisely, we have the following explicit characterization (see e.g.~\cite[Section 6.2.2.1]{GTV} and \cite{DM-2015-halfline}).

\begin{equation}\label{repres_point}
-\Delta_{\alpha}f=-f'',\qquad\mathcal{D}(-\Delta_{\alpha})=\left\{f\in L^2(\R^+)\,\Bigg|
\begin{array}{c}
f,f'\in W^{1,1}(\R^+),\\
f''\in L^2(\R^+),\\
\alpha f(0)=f'(0)
\end{array}
\right\}.
\end{equation}

Analogous models in $\R^d$ have been intensively studied, since the first rigorous attempt by Berezin and Faddeev \cite{Berezin-Faddeev-1961}, and subsequent characterizations by many other authors, see e.g.~\cite{Albeverio_Brzesniak-Dabrowski-1995,Albeverio-Fenstad-HoeghKrohn-1979_singPert_NonstAnal,albverio-kurasov-2000-sing_pert_diff_ops,DMSY-2017,Grossmann-HK-Mebkhout-1980,Grossmann-HK-Mebkhout-1980_CMPperiodic,MO-2016,MOS-2018,Posilicano2000_Krein-like_formula,Scarlatti-Teta-1990} (we refer to the monograph \cite{albeverio-solvable}, the surveys \cite{Albeverio-Figari,DFT-brief_review_2008}, and references therein for a comprehensive overview). 

In particular, it is well-known that the  operator $-\Delta|_{\mathcal{C}^{\infty}_c(\R^d\setminus\{0\})}$ is essentially self-adjoint in dimension $d\geq 4$, while it admits non-trivial self-adjoint extensions in dimension $d=1,2,3$. More precisely, in dimension $d=2,3$, there is a one-parameter family $\{-\Delta_{\alpha}\}_{\alpha\in\R\cup\{\infty\}}$ of self-adjoint extensions, analogously to the half-line case, while in dimension $d=1$ there is a richer family of extensions, including the so-called $\delta$ and $\delta'$ interactions.

A particularly relevant question is to understand whether physically meaningful operators of the form \eqref{formal_operator} converge (in a suitable sense), as $\eps\to 0$ and $\lambda\to\infty$, to a Schr\"odinger operator with point interaction. Such question has been widely studied in the Euclidean case $\R^d$ (see e.g.~\cite{albeverio-solvable,AHK-1981-JOPTH,GH10,GH13,MS-2018-shrinking-and-fractional} and references therein). The half-line case has been investigated in \cite{DM-2015-halfline,Seba}, by means of classical techniques. In \cite{DM-2015-halfline} the authors discussed the sub-critic scaling $\lambda\sim\eps^{-2+\delta}$, showing that the singular perturbation preserves the Dirichlet boundary condition. The critical case $\lambda\sim\eps^{-2}$ has been addressed in \cite{Seba}, where the author shows that the singular perturbation can produce a non-trivial boundary condition in the limit as $\eps\to 0$, under suitable spectral assumptions on the potential. Analogous results have been proven within the framework of nonstandard analysis, in the case of a square potential \cite{Albeverio-Fenstad-HoeghKrohn-1979_singPert_NonstAnal,Nelson}.

In particular, in their pioneering work~\cite{Albeverio-Fenstad-HoeghKrohn-1979_singPert_NonstAnal}, Albeverio, Fenstad, and H\o{}egh-Krohn 
have shown that the self-adjoint operator $-\Delta_\alpha$ on $L^2(\R^+)$ can be associated with the nonstandard self-adjoint operator
\begin{equation}
H_{\alpha} = -\Delta_D + \lambda_{\alpha} \mathbbm{1}_{[0,\varepsilon]}
\end{equation}
on $^*\! L^2(\mathbb{R}^+)$ (the ‘‘nonstandard extension of $L^2(\R^+)$"), where $\varepsilon$ is a fixed positive infinitesimal, $\mathbbm{1}_{[0,\varepsilon]}$ is the indicator function for the nonstandard interval $[0,\eps]$, and $\lambda_\alpha$ is an infinite coupling constant, of the form
\begin{equation}
\lambda_\alpha = -\Bigl( k + \frac{1}{2} \Bigr) \frac{\pi^2}{\varepsilon^2} + \frac{2}{\varepsilon} \alpha + \beta,
\end{equation}
for arbitrary $k\in\mathbb{N}$ and $\beta\in\mathbb{R}$. As a consequence, they also derived a similar result in the Euclidean space $\R^3$, using the factorization of $-\Delta_{\alpha}$ with respect to the angular momentum decomposition of $L^2(\R^3)$.

The main goal of this paper is to extend the results in \cite{Albeverio-Fenstad-HoeghKrohn-1979_singPert_NonstAnal} to a large class of compactly supported potentials.

Let us state explicitly that we do not assume the readers to have any prior knowledge of nonstandard analysis: analogously to \cite{Albeverio-Fenstad-HoeghKrohn-1979_singPert_NonstAnal}, our main results, and their proofs, use only some of the most basic nonstandard definitions and notions, which we will reintroduce in Section \ref{NSA}. The rest of the paper is structured as follows. In Section \ref{MR}, we state our main results, Theorem \ref{main} and Corollary \ref{cor:app}, whose proofs are postponed to Section \ref{PMR}. In Section \ref{linear}, we discuss explicitly the case of potentials that are linear in a neighbourhood of the origin. In Section \ref{Conclusions}, we draw some final remarks and perspectives.

\section{Nonstandard concepts}\label{NSA}

Nonstandard analysis belongs to the field of non-Archimedean mathematics, i.e.~to the study of structures containing infinite and infinitesimal quantities. The structure we will use is that of non-Archimedean fields\footnote{More in general, our problem could be approached using non-Archimedean rings, as it is done, e.g.~in Colombeau theory. However, as in this paper we will focus only on nonstandard analysis, we prefer not to give the most general possible definitions.}:

\begin{definition}
Let $\mathbb{K}$ be an infinite totally ordered field. An element $\xi \in \mathbb{K}$ is:

\begin{itemize}
\item infinitesimal if, for all positive $n\in \mathbb{N}$, $|\xi|<\frac{1}{n}$;
\item finite if there exists $n\in \mathbb{N}$ such that $|\xi |<n$;
\item infinite if, for all $n\in \mathbb{N}$, $|\xi |>n$ (equivalently, if $\xi $ is not finite).
\end{itemize}
We say that $\mathbb{K}$ is non-Archimedean if it contains an infinitesimal $\xi \neq 0$, and that $\mathbb{K}$ is superreal if it properly extends $\mathbb{R}$.
\end{definition}

By the completeness of $\mathbb{R}$, it follows that any infinite totally ordered superreal field is automatically non-Archimedean (whilst the converse is false: not all non-Archimedean fields are superreal).

Infinitesimals can be used to formalize the notion of closeness:

\begin{definition}
\label{def infinite closeness} We say that two numbers $\xi ,\zeta \in {\mathbb{K}}$ are infinitely close if $\xi -\zeta $ is infinitesimal. In this case we write $\xi \sim \zeta $.
\end{definition}

The relation $\sim$ is clearly an equivalence relation; in the case of the superreal fields, the following Theorem, whose proof is just a simple application of the completeness of $\mathbb{R}$, see e.g.~\cite[Theorem 5.6.1]{Goldblatt}, completely characterizes the quotient space of finite elements of $\mathbb{K}$, making precise the relationship between $\mathbb{K}$ and $\mathbb{R}$:

\begin{theorem}
If $\mathbb{K}$ is a superreal field, every finite number $\xi \in \mathbb{K}$ is infinitely close to a unique real number $r\sim \xi $, called the \textbf{standard part} of $\xi$.
\end{theorem}

Given a finite $\xi\in\mathbb{K}$, we will denote its standard part by $\st(\xi)$. Moreover, with a small abuse of notation, we will also write $\st(\xi )=+\infty $ (resp.~$\st(\xi )=-\infty $) if $\xi \in \mathbb{K}$ is a positive (resp.~negative) infinite number.

In the sequel, we will work also with complex numbers. The analogue of complex numbers for superreal fields can be easily introduced by considering
\begin{equation*}
\mathbb{K}+i\mathbb{K},
\end{equation*}%
namely, a field of numbers of the form 
\begin{equation*}
a+ib,\ a,b\in \mathbb{K}.
\end{equation*}

In this paper, the superreal field we use is any field of hyperreal numbers $^{\ast}\!\mathbb{R}$ coming from nonstandard analysis. Readers interested in a comprehensive introduction to nonstandard analysis are referred to \cite{Goldblatt}; here, we will only recall (in a semi-formal way) the very basic facts we need in the following sections, following the simplified presentation of nonstandard methods of \cite{BDNF}.

For most applications, nonstandard methods consist of just two ingredients:
\begin{itemize}
\item two mathematical universes\footnote{We will use this notion informally here, referring to \cite{Goldblatt} or any other book on nonstandard analysis for a precise formalization; informally, a mathematical universe $\mathbb{U},\mathbb{V}$ is a collection that is closed with respect to the basic set operations and includes the reals and any ‘‘mathematical object" that can be constructed starting with the reals, e.g.~functions, sets of functions, functionals, sets of functionals, sets of sets of functions and so on.} $\mathbb{U},\mathbb{V}$;
\item a star map $\ast:\mathbb{U}\rightarrow\mathbb{V}$ that associates to every object $x\in \mathbb{U}$ an object $^{\ast}\!x\in\mathbb{V}$, called its {\itshape hyperextension}, which satisfies the transfer principle.
\end{itemize}
As always in applications of nonstandard methods to analysis, we will also assume that $\mathbb{R}\subset\,\!^{\ast}\!\mathbb{R}$ and, for all real numbers $r\in\mathbb{R}$, we identify $r$ with its hyper-extension, namely we let $^{\ast}r=r$. Moreover, where there is no danger of confusion, we also identify a real function $f$ with its hyper-extension $^*\! f$.

The transfer principle states the following:

\begin{transfer*} Let $P\left(a_{1},\dots,a_{n}\right)$ be an elementary property of the objects $a_{1},\dots,a_{n}$. Then $P\left(a_{1},\dots,a_{n}\right)$ holds if and only if $P\left(\,\!^{\ast}a_{1},\dots,\,\!^{\ast}a_{n}\right)$ holds. \end{transfer*}

Formally, {\itshape elementary properties} means first order properties (in some language) with bounded quantifiers, see \cite[Section 4.4]{Chang} for details. For our semi-formal presentation, it is sufficient to know that a property is elementary if:

\begin{itemize}
\item it involves only the usual logic connectives (and, or, if, then, not), quantifiers (there exists, for all), and the basic notions of function, the value of a function at a given point, relation, domain, codomain, ordered n-tuple, the i-th component of an ordered
tuple, and membership;
\item all quantifiers are bounded by some set, namely, they appear in the form $\forall x\in X$ or $\exists x\in X$, where $X$ is some set.
\end{itemize}

\begin{fact*} We will apply the transfer principle only to the following properties, that are all elementary:
\begin{itemize} 
\item ``for all $\varepsilon\in\R^+$ and $\lambda\in\mathbb{R}$, the operator $H_{\lambda,\varepsilon}:=-\Delta_{D}+\lambda V\left(\frac{x}{\varepsilon}\right)$ admits a unique self-adjoint realization on $L^{2}(\mathbb{R}^{+})$, which is bounded from below in $\mathbb{R}$. Moreover, for all $z\in\mathbb{C}\setminus\mathbb{R}$, for all $g\in L^{2}\left(\mathbb{R}^{+}\right)$, 
$$((H_{\lambda,\eps}-z)^{-1}g)(x)=\int_{\mathbb{R}^{+}} G_{\lambda,\eps,z}(x,y)g(y)dy\mbox{'',}$$ 
which by transfer becomes:\\

\noindent``for all $\varepsilon\in\,\!^{\ast}\R^+$ and $\lambda\in\,\!^{\ast}\!\mathbb{R}$, the operator $H_{\lambda,\varepsilon}:=-\Delta_{D}+\lambda V\left(\frac{x}{\varepsilon}\right)$ admits a unique self-adjoint realization on $^{\ast}(L^{2}(\mathbb{R}^{+}))$, which is bounded from below in $^{\ast}\!\mathbb{R}$ (notice that this hence holds for any $\varepsilon$ positive infinitesimal and for every $\lambda$). Moreover, for all $z\in\,\!^{\ast}\!\left(\mathbb{C}\setminus\mathbb{R}\right)$, for all $g\in\,\!^{\ast}\!\left(L^{2}\left(\mathbb{R}^{+}\right)\right)$, 
$$((H_{\lambda,\eps}-z)^{-1}g)(x)=\int_{^{\ast}\!\mathbb{R}^{+}} G_{\lambda,\eps,z}(x,y)g(y)dy\mbox{'';}$$
\item ``for all $\omega,\theta\in\R$ and $\eps\in\R^+$, the operator $-\Delta+\frac{\theta+\omega\eps}{\eps^2}W\left(\frac{x}{\eps}\right)$ admits a unique self-adjoint realization on $L^2(\R^3)$'', which by transfer becomes: ``for all $\omega,\theta\in\,\!^*\R$ and $\eps\in\!\mathstrut^*\R^+$, the operator $-\Delta+\frac{\theta+\omega\eps}{\eps^2}W\left(\frac{x}{\eps}\right)$ admits a unique self-adjoint realization on $^*\! L^2(\R^3)$''; 
\item Taylor's formula;
\item nonstandard formulation of equations \eqref{finde} and \eqref{eq:res_ker}.
\end{itemize}
\end{fact*}

\begin{definition} A model of nonstandard methods is a triple $\langle \mathbb{U},\mathbb{V},\ast\rangle$ where $\mathbb{U},\mathbb{V}$ are mathematical universes and $\ast:\mathbb{U}\rightarrow\mathbb{V}$ is a star map that satisfies the transfer principle.\end{definition} 

To avoid confusion, we will use the term {\itshape standard} when referring to objects in $\mathbb{U}$, e.g.~when talking about real functions.

From now on, we consider given a model of nonstandard methods. We will work in the non Archimedean field $^{\ast}\!\mathbb{R}$. We will use also the following notations\footnote{Which slightly differ those more usual in nonstandard analysis, but are closer to those of classical analysis.}:
\begin{itemize}
\item given $x,y\in\,\!\!^*\R$, we write $x=O(y)$ if there exists $C\in\R^+$ such that $|x|\leq C|y|$;
\item we write $x=\Theta(y)$ if $x=O(y)$ and $y=O(x)$;
\item we write $x=o(y)$ if there exists $\eps$ infinitesimal such that $|x|\leq\eps |y|$; in particular, $x=o(1)$ means that $x$ is infinitesimal;
\item let $\Omega$ be an open subset of $\mathbb{R}^{n}$, and let $f:\!^{\ast}\Omega\rightarrow\!^{\ast}\mathbb{R}$; we let $\st(f)$ be the function $\st(f):\!\,^{\ast}\Omega\rightarrow\mathbb{R}\cup\{-\infty,+\infty\}$ such that for every $x\in\,\!^*\Omega$, $\st(f)(x):=\st\left(f(x)\right)$.
\end{itemize}

\section{Main results}\label{MR}

In this Section we state the main results of this paper, postponing their proofs to Section \ref{PMR}. 

We will use the following notations:

\begin{itemize}
\item for $z\in\C\setminus\R_{\geq 0}$, we write $\sqrt{z}$ to denote the branch of the square root such that $\RE\sqrt{z}>0$;
\item given $x\in\R^d$, we let $\langle x\rangle:=\sqrt{1+|x|^2}$;
\item given any two functions $f,g\in\mathcal{C}^1(I)$, for some interval $I\subseteq\R$, we let $W(f,g):=fg'-f'g$ be the Wronskian of $f$ and $g$;
\item for any set $X\subseteq\R^d$, we denote by $\mathbbm{1}_X$ the characteristic function of $X$;
\item given a function $f\in L^1_{\mathrm{loc}}(\R^3)$, we let $\mathbb{P}_0f$ be its spherical mean, given by
$$(\mathbb{P}_0f)(x)=\frac{1}{4\pi}\int_{\mathbb{S}^2}f\left(|x|\omega\right)d\omega,\quad x\in\R^3.$$
\end{itemize}

Let us define the operators we want to study. We start by fixing a compactly supported potential $V\in L^{\infty}(\R^+,\R)$. For any given $\eps\in\R^+$ and $\lambda\in\R$, it is well known \cite{Reed-Simon-II,Sch} that the operator
\begin{equation}\label{def_H}
H_{\lambda,\eps}:=-\Delta_D+\lambda V\left(\frac{x}{\eps}\right)
\end{equation}
admits a unique self-adjoint realization on $L^2(\R^+)$, which is bounded from below. There are more general conditions on $V$ which ensure the self-adjointness of $H_{\lambda,\eps}$, see e.g.~the discussion in \cite{DM-2015-halfline}. For the sake of concreteness, we do not discuss here the full general case, even though our arguments can be easily adapted.

Now, for any given $\eps\in\,\!^*\R^+$ and $\lambda\in\,\!^*\R$, as observed in Section \ref{NSA} the transfer principle guarantees that $H_{\lambda,\eps}$ defines a self-adjoint, bounded below operator on $^*\! L^2(\R^+)$. We address the question of whether $H_{\lambda,\eps}$ can be restricted to a standard operator. Let us first provide a precise definition.

\begin{definition}\label{de:ns}
Let $\Omega$ be an open subset of $\R^d$. Let $A$ be a self-adjoint operator on $^*\! L^2(\Omega)$. We say that $A$ is \emph{near standard} if there exists a self-adjoint operator $B$ on $L^2(\Omega)$ such that, for every $f\in L^2(\Omega)$ and for every $z\in\C\setminus\R$, we have
$$\st\big((A-z)^{-1}\,^*\! f\big)\,|_{\Omega}=(B-z)^{-1}f.$$ 
In this case, we write $\st(A)=B$.
\end{definition}

Our main result is the following.

\begin{theorem}\label{main}
Let us fix a potential $V\in L^{\infty}(\R^+,\R)$, with $supp(V)\subseteq[0,M]$ for some $M>0$, and let $\eps,\lambda\in\mathstrut^*\R$, with $\eps$ positive and infinitesimal and $|\lambda|=O(\eps^{-2})$. Then the self-adjoint operator $H_{\lambda,\eps}$ is near standard. 

Moreover, given $\theta\in\R$, let $\psi_{\theta}\in W_{\mathrm{loc}}^{2,\infty}(\R^+)\hookrightarrow\mathcal{C}^1(\R^+)$ be a non-zero solution to $(-\Delta_D+\theta V)\psi_{\theta}=0$, and consider the set
$$\Upsilon:=\{\theta\in\R\,|\,\psi'_{\theta}(M)=0\}.$$
If $\lambda$ has the form
$$\lambda=\frac{\theta}{\eps^2}+\frac{\omega}{\eps}+o(\eps^{-1}),\quad\theta\in\Upsilon,\,\omega\in\R,$$
then $\st(H_{\lambda,\eps})=-\Delta_{\alpha}$, with
\begin{equation}\label{alphafinale}
\alpha=\frac{\omega}{\psi_{\theta}(M)^2}\int_0^MV(t)\psi_{\theta}^2(t)dt.
\end{equation}
For all the other choices of $\lambda$, we have $\st(H_{\lambda,\eps})=-\Delta_{D}$.
\end{theorem}

\begin{remark}\label{deca} Notice that the hypothesis $|\lambda|=O\left(\varepsilon^{-2}\right)$ entails that $\lambda$ has the form $\theta\eps^{-2}+S(\eps)$, where $S(\eps)=o(\eps^{-2})$ and $\theta\in\mathbb{R}$. In fact, by our hypothesis we have
\[I_{\lambda,\varepsilon}:=\big\{C\in\mathbb{R}_{\geq 0}\mid \eps^2\,|\lambda|\leq C\big\}\neq \emptyset,\] and the thesis is reached by taking $\theta=\operatorname{sgn}(\lambda)\cdot\inf I_{\lambda,\varepsilon}$.
\end{remark}
\begin{remark}
Observe that, if $\theta\in\Upsilon$, we necessarily have $\psi_{\theta}(M)\neq 0$, whence $\alpha\neq +\infty$. This implies that $\st(H_{\lambda,\eps})$ is a Schr\"odinger operator with a non-trivial point interaction at the origin. Moreover, since
$$\int_0^MV(t)\psi_{\theta}^2(t)dt=-\frac{1}{\theta}\int_0^M(\psi_{\theta}')^2(t)dt\neq 0,$$
we deduce that, by varying the parameter $\omega$, we obtain all the possible values $\alpha\in\R$. 

Furthermore, $\alpha$ generally depends not only on $\omega$, but on $\theta$ as well through $\psi_\theta$. This should be compared with the result in \cite{Albeverio-Fenstad-HoeghKrohn-1979_singPert_NonstAnal}, as the dependence on $\theta\in\Upsilon$ vanishes for the square potential $V = \mathbbm{1}_{[0,1]}$: in this case, in fact, one has $\alpha = \frac{\omega}{2}$.
\end{remark}

As a consequence of Theorem \ref{main}, we deduce a nonstandard formulation of the general approximation theorem for Schr\"odinger operators with point interactions in $\R^3$ \cite{albeverio-solvable,MS-2018-shrinking-and-fractional}, in the case of radial, compactly supported potentials. 

In order to introduce the result, let us first recall some well-known facts about delta-like interactions in three dimensions (see e.g.~\cite[Chapter I.1]{albeverio-solvable}). As anticipated in the Introduction, the symmetric operator $-\Delta|_{\mathcal{C}^{\infty}_c(\R^3\setminus\{0\})}$ has a one-parameter family of self-adjoints extensions, that we denote by $\big\{\!-\!\Delta^{(3)}_{\alpha}\big\}$, $\alpha\in\R\cup\{\infty\}$, in order to avoid confusion with the 1D case. The extension with $\alpha=\infty$ coincides with the free Laplacian  on $\R^3$. Let us consider the angular momenta decomposition 
\begin{equation}\label{finde}
L_{\mathrm{loc}}^2(\mathbb{R}^3)=\bigoplus_{\ell=0}^{\infty}U^{-1}L_{\mathrm{loc}}^2(\mathbb{R}^+,rdr)\otimes \langle Y_{\ell,-\ell},\ldots ,Y_{\ell,\ell}\rangle,
\end{equation}
where $(Uf)(r)=rf(r)$, and $Y_{\ell,m},\,\ell\in\mathbb{N},\,m=0,\pm 1,\ldots, \pm\ell$, are the spherical harmonics on $L^2(S^2)$. With respect to this decomposition, the operator $-\Delta^{(3)}_{\alpha}$ writes as
\begin{equation}\label{de:se}
-\Delta^{(3)}_{\alpha}=\Big(-\Delta_{\alpha}\oplus\bigoplus_{\ell=1}^{\infty}U^{-1}H_{\ell}U\Big)\otimes 1,
\end{equation}
where $H_{\ell}$, $\ell\geq 1$, are self-adjoint operators on $L^2(\R^+,rdr)$, independent on $\alpha$. Explicitly, we have
\begin{equation}\label{HL}
\begin{gathered}
\mathcal{D}(H_{\ell})=\{f\in L^2(\R^+)\,|\,f,f'\in W_{\mathrm{loc}}^{1,1}(\R^+), \; -d_r^2f + \ell(\ell+1)r^{-2}f \in L^2(\R^+)\},\\
\quad H_{\ell} f=-d_r^2f + \ell(\ell+1)r^{-2}f.
\end{gathered}
\end{equation}
Identity \eqref{de:se} tells us that $-\Delta^{(3)}_{\alpha}$ diagonalizes with respect to decomposition \eqref{finde}, and that it coincides with $-\Delta$ after projecting out the subspace of radial functions. 

Next, let us fix a radial, compactly supported potential $W\in L^{\infty}(\R^3,\R)$, and consider the operator $-\Delta+W$, which admits a unique self-adjoint realization on $L^2(\R^3)$ (see e.g.~\cite{Sch}). We recall the notion of zero energy resonance for $-\Delta+W$ (see e.g.~\cite{MS-2018-shrinking-and-fractional} and references therein), which corresponds to the presence of a suitable distributional eigenfunction.
\begin{definition}
The operator $-\Delta+W$ is zero energy resonant if there exists a function $\Psi\in L^2(\R^3,\langle x\rangle^{-1-\delta}dx)\setminus L^2(\R^3)$, for any $\delta>0$, such that
$$(-\Delta+W)\Psi=0$$
as a distributional identity. The function $\Psi$ is called a zero-energy resonance.
\end{definition}

Now, for every $\omega,\theta\in\R$, and $\eps\in\mathstrut^*\R$ positive and infinitesimal, consider the operator
$$A_{\eps,\theta,\omega}:=-\Delta+\frac{\theta+\omega\eps}{\eps^2}W\Big(\frac{x}{\eps}\Big),$$
which is self-adjoint on $^*\! L^2(\R^3)$ by means of the transfer principle, as observed in Section \ref{NSA}. That $A_{\eps,\theta,\omega}$ is near standard is a direct consequence of Theorem \ref{main}, but we can say more: $\st\left(A_{\eps,\theta,\omega}\right)$ is non-trivial if and only if $-\Delta+\theta W$ is zero energy resonant, in a sense that is made precise by the following Corollary. 
\begin{corollary}\label{cor:app}
Let us fix a radial, compactly supported potential $W\in L^{\infty}(\R^3,\R)$, and write $W(|x|):=V(x)$, where $\supp V\subseteq [0,M]$ for some $M>0$. The self-adjoint operator $A_{\eps,\theta,\omega}$ is near standard. Moreover, we have the following dichotomy.
\begin{itemize}
\item[(i)] If $-\Delta+\theta W$ is not zero-energy resonant, then $\st(A_{\eps,\theta,\omega})=-\Delta$.
\item[(ii)] If $-\Delta+\theta W$ has a zero-energy resonance $\Psi_{\theta}$, then $\st(A_{\eps,\theta,\omega})=-\Delta^{(3)}_{\alpha}$, with 
 $$\alpha=\frac{\omega}{\psi_{\theta}(M)^2}\int_0^MV(t)\psi_{\theta}^2(t)dt\neq\infty,$$
where the function $\psi_{\theta}:\R^+\to\R$ is defined by
\begin{equation}\label{psi_ris}
\psi_{\theta}(|x|):=|x|(\mathbb{P}_{0}\Psi_{\theta})(x
),\quad x\in\R^3.
\end{equation}
\end{itemize}
\end{corollary}

We will see in the proof of Corollary \ref{cor:app} that the coupling parameters $\theta\in\Upsilon(V)$ are exactly those for which the unscaled operator $-\Delta+\theta W$ has a zero energy-resonance.

\section{Proof of the main results}\label{PMR}
In this Section we prove our main results, namely Theorem \ref{main} and Corollary \ref{cor:app}, by adapting the nonstandard techniques of \cite{Albeverio-Fenstad-HoeghKrohn-1979_singPert_NonstAnal} to our more general case.

Let us start with some preliminary discussion. Without loss of generality, we can assume that $\supp V\subseteq[0,1]$. For any fixed $z\in\C$, and for any given $\theta,\gamma\in\R$, we consider the unique solution $\psi_{\theta,\gamma}\in W_{\mathrm{loc}}^{2,\infty}(\R^+)\hookrightarrow\mathcal{C}^1(\R^+)$ to the Cauchy problem
\begin{equation}\label{depending-map}
\begin{cases}
-\psi_{\theta,\gamma}''+(\theta V-\gamma z)\psi_{\theta,\gamma}=0,\\
\psi_{\theta,\gamma}(0)=0\\
\psi_{\theta,\gamma}'(0)=1.
\end{cases}
\end{equation}
In particular, $\psi_{\theta}:=\psi_{\theta,0}$ is a non-zero solution to $(-\Delta_D+\theta V)\psi_{\theta}=0$. Moreover, it is easy to show that the map $(\theta,\gamma)\mapsto\psi_{\theta,\gamma}$ belongs to $\mathcal{C}^1(\R^2, W^{2,\infty}_{\mathrm{loc}}(\R^+))$.

Next, for every fixed $\eps\in(0,1]$ and $\lambda\in\R$, let us consider the self-adjoint operator $H_{\lambda,\eps}$ on $L^2(\R^+)$, defined by \eqref{def_H}. In order to characterize the resolvent map $(H_{\lambda,\eps}-z)^{-1}$, for $z\in\C\setminus\R$, we analyze the following equation:
\begin{equation}\label{Schr}
-\phi''(x)+\lambda V\Big(\frac{x}{\eps}\Big)\phi(x)-z\phi(x)=0,\quad x\in\R^+.
\end{equation}
Let us denote by $u_{\lambda,\eps},v_{\lambda,\eps}\!\in \mathcal{C}^{1}(0,\eps)$ the solutions to \eqref{Schr} on $[0,\eps]$, respectively with $u_{\lambda,\eps}(0)=0$, $u'_{\lambda,\eps}(0)=1$ and $v_{\lambda,\eps}(1)=0$, $v'_{\lambda,\eps}(0)=0$. Observe that, for $x\in[0,\eps]$,
\begin{equation}\label{scaling}
u_{\lambda,\eps}(x)=\eps\psi_{\eps^2\lambda,\eps^2}\Big(\frac{x}{\eps}\Big).
\end{equation}

Let $\phi_1,\phi_2\in\mathcal{C}^1(\R^+)$ be solutions to \eqref{Schr}, respectively with $\phi_1(0)=0$ and $\phi_2(x)\to 0$ as $x\to\infty$. Explicitly, we consider
\begin{equation}\label{phiuno}
\phi_1(r):=\begin{cases}
u_{\lambda,\eps}(x),&x\in[0,\eps]\\
ae^{\sqrt{z}x}+be^{-\sqrt{z}x},&x>\eps,
\end{cases}
\end{equation}

\begin{equation}\label{phidue}
\phi_2(r):=\begin{cases}
cu_{\lambda,\eps}(x)+dv_{\lambda,\eps}(x),&x\in[0,\eps]\\
e^{-\sqrt{z}x},&x>\eps,
\end{cases}
\end{equation}
where the constants $a,b,c,d\in\C$ are determined by imposing the differentiability of $\phi_1$ and $\phi_2$ at $x=\eps$. For a later convenience, we write down the expressions for $a$ and $b$:
\begin{equation}\label{expr:ab}
a=\frac{1}{2}e^{-\sqrt{z}\eps}\left(u_{\lambda,\eps}(\eps)+\frac{u_{\lambda,\eps}'(\eps)}{\sqrt{z}}\right),\quad b=\frac{1}{2}e^{\sqrt{z}\eps}\left(u_{\lambda,\eps}(\eps)-\frac{u_{\lambda,\eps}'(\eps)}{\sqrt{z}}\right).
\end{equation}

By Sturm--Liouville theory, the Wronskian $K:=W(\phi_1,\phi_2)$ is independent on $x$, and a direct computation shows that actually $K=-2a\sqrt{z}$. In addition, since $z\not\in\sigma(H_{\lambda,\eps})$, we have that $K\neq 0$, and the integral kernel of $(H_{\lambda,\eps}-z)^{-1}$ is given explicitly by
\begin{equation}\label{eq:res_ker}
G_z(x,y)=\frac{1}{K}\,\begin{cases}
\phi_1(x)\phi_2(y),&x\leq y,\\
\phi_1(y)\phi_2(x),&x\geq y.
\end{cases}
\end{equation}

Let us fix now $\eps,\lambda(\eps)\in\,\!^*\R$, with $\eps$ positive and infinitesimal and $|\lambda|=O(\eps^{-2})$. Owing to the transfer principle, we have that $H_{\lambda,\eps}$ is a self-adjoint operator on $^*\! L^2(\R^+)$, and $G_z$ is the integral kernel of its resolvent. More precisely, for every\footnote{By transfer, this property would actually hold more in general for every $z\in\,\!^{\ast}\!\left(\C\setminus\R\right)$.} $z\in\C\setminus\R$ and for every $g\in\,\!^*\! L^2(\R^+)$, we have
$$ 
\big((H_{\lambda,\eps}-z)^{-1}g\big)(x)=\int_{^*\R^+}G_z(x,y)g(y)dy.
$$
In particular, for every $x\in\,\!\!^*\R$, with $x\geq\eps$, and for every $f\in L^2(\R^+)$, formulas \eqref{phiuno}, \eqref{phidue} and \eqref{eq:res_ker} yield
\begin{equation}\label{exp_kernel}
\begin{split}
&\big((H_{\lambda,\eps}-z)^{-1}\,^*\! f\big)(x)=\frac{e^{-\sqrt{z}x}}{K}\Bigg(\int_0^{\eps}u_{\lambda,\eps}(y)\,^*\! f(y)dy\,\\
&\quad+\int_{\eps}^x(ae^{\sqrt{z}y}+be^{-\sqrt{z}y})\,^*\! f(y)dy+(ae^{2\sqrt{z}x}+b)\int_x^{+\infty}e^{-\sqrt{z}y}\,^*\! f(y)dy\Bigg).
\end{split}
\end{equation}

In the following Lemma, we show the restriction to $\R^+$ of the resolvent map \eqref{exp_kernel} defines a standard operator.

\begin{lemma}\label{le:stare}
For every $z\in\C\setminus\R$, the map
$$\Phi_z:f\mapsto\st\big((H_{\lambda,\eps}-z)^{-1}\,^*\! f\big)|_{\R^+}$$
defines a bounded linear operator on $L^2(\R^+)$.
\end{lemma}

\begin{proof}
Analyzing the identity \eqref{exp_kernel}, and observing that $K^{-1}a=-(4z)^{-1/2}$ is finite, we easily deduce that it is sufficient to show that $K^{-1}b$ and $K^{-1}u_{\lambda,\eps}(y)$ are finite, for every $y\in[0,\eps]$. 

Let us write $\lambda=\theta\eps^{-2}+S(\eps)$, with $\theta\in\R$ and $S(\eps)=o(\eps^{-2})$ (such decomposition is indeed possible, see Remark \ref{deca}). We have the following expansions, which follow immediately from the representation \eqref{scaling}, Taylor's formula and the fact that the map $(\theta,\gamma)\mapsto\psi_{\theta,\gamma}$ belongs to $\mathcal{C}^1(\R^2, W_{\mathrm{loc}}^{2,\infty}(\R^+))$.
\begin{eqnarray}
\label{exp_G} u_{\lambda,\eps}(y) &=& \eps\psi_{\varepsilon^{2}\lambda,\varepsilon^{2}}(\eps^{-1}y)\\
\nonumber &=& \eps\psi_{\theta,0}(\eps^{-1}y)+o(\eps), \quad \forall\,y\in[0,\eps], \\
\label{exp_der_G} u'_{\lambda,\eps}(\eps) &=&\psi'_{\varepsilon^{2}\lambda,\varepsilon^{2}}(1) \\
\nonumber &=&\psi_{\theta,0}'(1)+\eps^2S(\eps)\frac{\partial\mathcal{G} }{\partial \theta}(\theta,0)+O(\max\{(\eps^2S(\eps))^2,\eps^2\}),
\end{eqnarray}
where we introduced the $\mathcal{C}^1$-map
$$\mathcal{G}:\R\times\R\to\R,\quad\mathcal{G}(\lambda,\gamma)=\psi_{\lambda,\gamma}(1).$$

Now we distinguish two cases.\\

\textbf{Case I:} $\theta\not\in\Upsilon$. It follows from \eqref{exp_G} that $u_{\lambda,\eps}(y)=O(\eps)$ for every $y\in[0,\eps]$. Moreover, $u_{\lambda,\eps}'(\eps)=\Theta(1)$, as $\psi_{\theta,0}'(1)\neq 0$ (whilst all other summands in expansion (\ref{exp_der_G}) are infinitesimal). Hence, we deduce from \eqref{expr:ab} that $b$ is finite and $K$ is not infinitesimal. We conclude that $K^{-1}b$ and $K^{-1}u_{\lambda,\eps}(y)$, $y\in[0,\eps]$, are finite.\\

\textbf{Case II:} $\theta\in\Upsilon$. Since $\psi'_{\theta,0}(1)=0$, expansion \eqref{exp_der_G} reduces to
\begin{equation}\label{eq:nuova}
u'_{\lambda,\eps}(\eps)=\eps^2S(\eps)\frac{\partial\mathcal{G} }{\partial \theta}(\theta,0)+O(\max\{(\eps^2S(\eps))^2,\eps^2\}).
\end{equation}
Let us compute explicitly $(\partial_{\theta}\mathcal{G})(\theta,0)$. We have $(\partial_{\theta}\mathcal{G})(\theta,0)=g'(1)$, where $g$ satisfies
\begin{equation*}
-g''+\theta Vg+V\psi_{\theta}=0,
\end{equation*}
with initial conditions $g(0)=g'(0)=0$. Let $\tilde{\psi}_{\theta}\in\mathcal{C}^1(\R^+)$ be the solution to the Cauchy problem 
\begin{equation*}
\begin{cases}
-\tilde{\psi}_{\theta}''+\theta V\tilde{\psi}_{\theta}=0,\\
\tilde{\psi}_{\theta}(0)=1\\
\tilde{\psi}_{\theta}'(0)=0.
\end{cases}
\end{equation*}
Using $W(\psi_{\theta},\tilde{\psi}_{\theta})=-1$, we can apply the variation of constants to get
$$g(x)=\Big(\int_0^xV(t)\tilde{\psi}_{\theta}(t)\psi_{\theta}(t)dt\Big)\psi_{\theta}(x)-\Big(\int_0^xV(t)\psi_{\theta}(t)^2dt\Big)\tilde{\psi}_{\theta}(x).$$
Taking into account the values of $\psi_{\theta},\tilde{\psi}_{\theta}$ at $x=0$ and the condition $\psi'_{\theta}(1)=0$, we obtain
\begin{eqnarray}
\nonumber \frac{\partial\mathcal{G} }{\partial \theta}(\theta,0) &=& g'(1) \\
\nonumber &=& -\Big(\int_0^1V(t)\psi_{\theta}(t)^2dt\Bigr)\tilde{\psi}'_{\theta}(1).
\end{eqnarray}
Finally, the identity $\psi_{\theta}(1)\tilde{\psi}'_{\theta}(1)=W(\psi_{\theta},\tilde{\psi}_{\theta})=-1$ yields
\begin{equation}\label{eq:derivative_G}
\frac{\partial\mathcal{G} }{\partial \theta}(\theta,0)=\frac{1}{\psi_{\theta}(1)}\Big(\int_0^1V(t)\psi_{\theta}(t)^2dt\Big).
\end{equation}
In particular, as $\theta V(t)\tilde{\psi}_{\theta}^{2}=\tilde{\psi}_{\theta}^{\prime\prime}\tilde{\psi}_{\theta}$, an integration by parts yields
\begin{equation}\label{gnonzero}
\frac{\partial\mathcal{G} }{\partial \theta}(\theta,0)=-\frac{1}{\theta\psi_{\theta}(1)}\int_0^1(\psi_{\theta}')^2(t)dt\neq 0.
\end{equation}
Next, observe that $\psi_{\theta,0}(1)\neq 0$, otherwise we would have $\psi_{\theta,0}\equiv 0$. Owing to formulas \eqref{expr:ab}, \eqref{exp_G}, \eqref{eq:nuova}, and using that $\psi_{\theta,0}(1)\in\R\setminus\{0\}$, $\partial_{\theta}\mathcal{G}(\theta,0)\in\R\setminus\{0\}$, $\sqrt{z}\not\in\R$, we obtain
\begin{eqnarray}
\nonumber K &=& -2a\sqrt{z}\\
\nonumber &=& -e^{-\sqrt{z}\eps}\big(\sqrt{z}u_{\lambda,\eps}(\eps)+u_{\lambda,\eps}(\eps)\big)\\
\nonumber &=&-e^{-\sqrt{z}\eps}\Big(\sqrt{z}\eps\psi_{\theta,0}(1)+o(\eps)+\eps^2S(\eps)\frac{\partial\mathcal{G} }{\partial \theta}(\theta,0)+O(\max\{(\eps^2S(\eps))^2,\eps^2\})\Big)\\
\nonumber &=&\Theta(\max\{\eps,\eps^2S(\eps)\}),
\end{eqnarray}
given that no cancellation can occur between the terms $\sqrt{z}\eps\psi_{\theta,0}(1)\not\in\R$ and $\eps^2S(\eps)(\partial_{\theta}\mathcal{G})(\theta,0)\in\R$. Analogously, we deduce that
\begin{eqnarray}
\nonumber b &=& \frac{1}{2}e^{\sqrt{z}\eps}\left(u_{\lambda,\eps}(\eps)-\frac{u_{\lambda,\eps}'(\eps)}{\sqrt{z}}\right)\\
\nonumber &=& \Theta(\max\{\eps,\eps^2S(\eps)\}).
\end{eqnarray}
It follows that $K^{-1}b$ is finite. Moreover, since $u_{\lambda,\eps}(y)=O(\eps)$ for $y\in[0,\eps]$, and $K^{-1}=\big(\Theta(\max\{\eps,\eps^2S(\eps)\})\big)^{-1}=O(\eps^{-1})$, we also obtain that $K^{-1}u_{\lambda,\eps}(y)$ is finite for every $y\in[0,\eps]$. The proof is complete.
\end{proof}

Observe now that, given $f\in\mathcal{C}^{\infty}_c(\R^+)$ and $z\in\C\setminus\R$, we have 
$$(H_{\lambda,\eps}-z)\,^*\! f=\,^*\! ((-\Delta-z)f),$$
or equivalently
$$^*\! f=(H_{\lambda,\eps}-z)^{-1}\,^*\! ((-\Delta-z)f).$$
Taking the standard part and restricting to $\R^+$ we obtain
\begin{equation}\label{moPhi}
\Phi_{z}\left((-\Delta-z)f\right)=f,\qquad\forall\,f\in \mathcal{C}^{\infty}_c(\R^+),
\end{equation}
where $\Phi_z$ is the bounded linear operator defined in Lemma \ref{le:stare}. It follows that $\Phi_{z}=(-\Delta_{\alpha(z)}-z)^{-1}$, for a suitable $\alpha(z)\in\R\cup\{\infty\}$. In order to conclude the proof of the main Theorem, it remains to show that $\alpha(z)$ is actually independent on $z$. To this aim, let us fix $z\in\C\setminus\R$ and an arbitrary $f\in L^2(\R^3)\setminus\{0\}$, and set $g_z:=(H_{\lambda,\eps}-z)^{-1}\,^*\! f$. By construction, $\st(g_z)|_{\R^+}$ belongs to $\mathcal{D}(-\Delta_{\alpha(z)})\setminus\{0\}$, and it follows by representation \eqref{repres_point} that $\alpha(z)=\st\big(g_z'(\eps)/g_z(\eps)\big)$. Owing to formula \eqref{exp_kernel}, a direct computation yields $\st\big(g_z'(\eps)/g_z(\eps)\big)=\st\big(u_{\lambda,\eps}'(\eps)/u_{\lambda,\eps}(\eps)\big)$, whence
\begin{equation}\label{alpha}
\alpha(z)=\st\big(u_{\lambda,\eps}'(\eps)/u_{\lambda,\eps}(\eps)\big).
\end{equation}
Using \eqref{alpha} and the expansions \eqref{exp_G}-\eqref{exp_der_G} we can determine $\alpha(z)$, showing that it is indeed independent on $z\in\C\setminus\R$. Again, we write $\lambda=\theta\eps^{-2}+S(\eps)$, with $\theta\in\R$ and $S(\eps)=o(\eps^{-2})$, and we distinguish two cases.\\

\textbf{Case I:} $\theta\not\in\Upsilon$.  Since $u'_{\lambda,\eps}(\eps)=\Theta(1)$ and $u_{\lambda,\eps}(\eps)=O(\eps)$, we get $\alpha(z)=\infty$ for every $z\in\C\setminus\R$. We conclude that $H_{\lambda,\eps}$ is near standard, in the sense of Definition \ref{de:ns}, with $\st(H_{\lambda,\eps})=-\Delta_D$.\\

\textbf{Case II:} $\theta\in\Upsilon$. We already observed that $\psi_{\theta,0}(1)\neq 0$ and $(\partial_{\theta}\mathcal{G})(\theta,0)\neq 0$, which in view of expansions \eqref{exp_G}-\eqref{exp_der_G} implies $u_{\lambda,\eps}(\eps)=\Theta(\eps)$ and $u'_{\lambda,\eps}(\eps)=\Theta(\eps^2S(\eps))$.
If $S(\eps)=\Theta(\eps^{-\gamma})$ for some $\gamma\in(1,2)$, then $u^{\prime}_{\lambda,\varepsilon}(\varepsilon)/u_{\lambda,\varepsilon}(\varepsilon)$ is of the order of $\varepsilon^{1-\gamma}$, which is infinite, hence by \eqref{alpha} we have that $\alpha(z)=\infty$ for every $z\in\C\setminus\R$. We conclude that  $H_{\lambda,\eps}$ is near standard, with $\st(H_{\lambda,\eps})=-\Delta_D$. 

If $S(\eps)=\omega\eps^{-1}+o(\eps^{-1})$, with $\omega\in\R$, then \eqref{alpha} yields
\begin{equation}\label{prel_a}
\alpha(z)=\frac{\omega}{\psi_{\theta,0}(1)}\Big(\frac{\partial\mathcal{G} }{\partial \theta}(\theta,0)\Big)\equiv\alpha.
\end{equation}
We deduce that $H_{\lambda,\eps}$ is near standard, with $\st(H_{\lambda,\eps})=-\Delta_{\alpha}$. Moreover, combining \eqref{prel_a} with \eqref{eq:derivative_G}, we obtain the desired expression \eqref{alphafinale} for $\alpha$. 

We have completed the proof of Theorem \ref{main}. 

We conclude this Section by proving Corollary \ref{cor:app}.

\begin{proof}[Proof of Corollary \ref{cor:app}]
Consider the operator 
$$H_{\eps,\theta,\omega}:=-\Delta_{D}+\frac{\theta+\omega\eps}{\eps^2}V\left(\frac{x}{\eps}\right),$$
which is near standard by means of Theorem \ref{main}, as $\theta$ and $\omega$ are finite\footnote{To apply Theorem \ref{main}, it would have been sufficient to assume that $\theta+\omega\varepsilon = O(1)$.}, with $\st(H_{\eps,\theta,\omega})=-\Delta_{\alpha}$ for a suitable $\alpha:=\alpha(\theta,\omega)\in\R\cup\{\infty\}$. 

Let us consider the nonstandard decomposition
\begin{equation}\label{findeno}
^\ast L^2(\mathbb{R}^3)=\bigoplus_{\ell=0}^{\infty}U^{-1}\,^\ast L^2(\mathbb{R}^+,rdr)\otimes \langle Y_{\ell,-\ell},\ldots ,Y_{\ell,\ell}\rangle,
\end{equation}
obtained by applying the transfer principle to \eqref{finde}. With respect to decomposition \eqref{findeno}, we have the factorization
\begin{equation}\label{pigra}
\big(A_{\eps,\theta,\omega}-z)^{-1}=\Big(\big(H_{\eps,\theta,\omega}-z)^{-1}\oplus
\bigoplus_{\ell=1}^{\infty}U^{-1}(\,^{\ast}H_{\ell}-z)^{-1}U\Big)\otimes 1,
\end{equation}
for every $z\in\C\setminus\R$, where $H_{\ell}$, $\ell\geq 1$, are the operators defined by \eqref{HL}.

Moreover, we define the extension map $\mathcal{E}_3:L^2(\R^3)\ni f\mapsto\,^*\! f\in\,\!^*\! L^2(\R^3)$, and the restriction map $\mathcal{R}_3:\,\!\!^*\! L^2(\R^3)\ni g\mapsto \st(g)|_{\R^3}$. Analogously, we define the extension map $\mathcal{E}_1:L^2(\R^+)\ni f\mapsto\,^*\! f\in\,\!^*\! L^2(\R^+)$, and the restriction map $\mathcal{R}_1:\,\!\!^*\! L^2(\R^+)\ni g\mapsto \st(g)|_{\R^+}$. Using \eqref{pigra} and the linearity of the extension and restriction maps, we obtain
\begin{equation*}\label{bidol}
\begin{split}
\mathcal{R}_3\big(&A_{\eps,\theta,\omega}-z)^{-1}\,\mathcal{E}_3
=\mathcal{R}_3\,\left(\Big(\big(H_{\eps,\theta,\omega}-z)^{-1}\oplus
\bigoplus_{\ell=1}^{\infty}U^{-1}(\,^{\ast}H_{\ell}-z)^{-1}U\Big)\otimes 1\right)\,\mathcal{E}_3\\
&=\Big(\big(\mathcal{R}_1\big(H_{\eps,\theta,\omega}-z)^{-1}\,\mathcal{E}_1\big)\otimes 1\Big)\oplus\mathcal{R}_3\left(
\bigoplus_{\ell=1}^{\infty}U^{-1}(\,^{\ast}H_{\ell}-z)^{-1}U\otimes 1\right)\mathcal{E}_3\\
&=\Big((-\Delta_{\alpha}-z)^{-1}\oplus
\bigoplus_{\ell=1}^{\infty}U^{-1}(H_{\ell}-z)^{-1}U\Big)\otimes 1,
\end{split}
\end{equation*}
for every $z\in\C\setminus\R$, as an identity with respect to the standard decomposition \eqref{finde}. Comparing the identity above with \eqref{de:se}, we deduce that $A_{\eps,\theta,\omega}$ is near standard, with $\st(A_{\eps,\theta,\omega})=-\Delta^{(3)}_{\alpha}$. It remains to prove the dichotomy.\\

(i) Assume that $-\Delta+\theta W$ is not zero energy resonant. Let $\psi_{\theta}\in\mathcal{C}^1(\R^+)$ be a non-zero solution to $(-\Delta_D+\theta V)\psi_{\theta}=0$, and suppose that $\theta\in\Upsilon(V)$, namely $\psi_{\theta}'(M)=0$. It follows that $\psi_{\theta}(x)=\psi_{\theta}(M)$ for $x\geq M$. Hence the function $\Psi_{\theta}$, defined by $\Psi_{\theta}(x)=|x|^{-1}\psi_{\theta}(|x|)$, belongs to $L^2(\R^3,\langle x\rangle^{-1-\delta})\setminus L^2(\R^3)$, for $\delta>0$, and satisfies $(-\Delta+\theta W)\Psi_{\theta}=0$, namely it is a zero-energy resonance for $-\Delta + \theta W$, yielding a contradiction. We deduce that $\theta\not\in\Upsilon(V)$, whence $\alpha=\infty$. \\

(ii) Assume that $\Psi_{\theta}$ is a zero-energy resonance for $-\Delta+\theta W$. Since $\Psi_{\theta}\in L^2_{\mathrm{loc}}(\R^3)$ and $W\in L^{\infty}(\R^3)$, the relation $(-\Delta+\theta W)\mathbb{P}_0\Psi_{\theta}=0$ implies that $\mathbb{P}_0\Psi_{\theta}\in H^2_{\mathrm{loc}}(\R^3)$. Then the Morrey-Sobolev embedding for radial functions (see e.g.~\cite[Proposition 1.1]{SYY}) guarantees that the function $\psi_{\theta}$, defined by \eqref{psi_ris}, belongs to $W^{1,\infty}_{\mathrm{loc}}(\R^+)$. Moreover,  since $H^2_{\mathrm{loc}}(\R^3)\hookrightarrow L^{\infty}_{\mathrm{loc}}(\R^3)$, the relation $(-\Delta+\theta W)\mathbb{P}_0\Psi_{\theta}=0$ yields also $\Delta\psi_{\theta}\in L^{\infty}_{\mathrm{loc}}(\R^+)$. Combining everything, we conclude that $\psi_{\theta}\in W^{2,\infty}_{\mathrm{loc}}(\R^+)$. 

Next, we consider the representation of $\Psi_{\theta}$ with respect to decomposition \eqref{finde}:
\begin{equation}\label{decophi}
\Psi_{\theta}=\big(U^{-1}\psi_{\theta}\otimes Y_{0}\big)\oplus\bigoplus_{\ell=1}^{\infty}\bigoplus_{m=-\ell}^{\ell}U^{-1}\psi_{\theta}^{(\ell,m)}\otimes Y_{\ell,m},
\end{equation}
for suitable $\psi_{\theta}^{(\ell,m)}\in L_{\mathrm{loc}}^2(\R^+)$. Owing to the relation $(-\Delta+\theta W)\Psi_{\theta}=0$, we obtain that $(-\Delta+\theta V)\psi_{\theta}=0$ and 
\begin{equation}\label{eq:lm}
(H_{\ell}+\theta V)\psi^{(\ell,m)}_{\theta}=0,\quad \ell\geq 1,\,m=-\ell,\ldots, \ell.
\end{equation}
Given that $\supp V\subseteq[0,M]$, equation \eqref{eq:lm} and the characterization \eqref{HL} yield
$$\psi_{\theta}^{(\ell,m)}(r)=A_{\ell,m}r^{\ell+1}+B_{\ell,m}r^{-\ell},\quad r\geq M,$$
for suitable $A_{\ell,m},B_{\ell,m}\in\R$. Since $\Psi_{\theta}\in L^2(\R^3,\langle x\rangle^{-1-\delta}dx)$, we necessarily have $A_{\ell,m}=0$, which implies
$$\left|\big(U^{-1}\psi_{\theta}^{(\ell,m)}\otimes Y_{\ell,m}\big)(x)\right|\lesssim \frac{1}{|x|^{\ell+1}},\quad |x|\geq M.$$
In particular, we have $U^{-1}\psi_{\theta}^{(\ell,m)}\otimes Y_{\ell,m}\in L^2(\R^3)$ for every $\ell\geq 1,\,m=-\ell,\ldots, \ell$. Since $\Psi_{\theta}\not\in L^2(\R^3)$, we deduce from \eqref{decophi} that $U^{-1}\psi_{\theta}\otimes Y_{0}\not\in L^2(\R^3)$, which guarantees that $\psi_{\theta}$ is not identically zero.

Summarizing so far, we have proved that $\psi_{\theta}\in W^{2,\infty}_{\mathrm{loc}}(\R^+)$ is a non-zero function, satisfying the equation $(-\Delta+\theta V)\psi_{\theta}=0$. Suppose now that $\psi_{\theta}'(M)=\beta\neq 0$. It would follow that $\psi_{\theta}(x)=\psi_{\theta}(M)+\beta x$ for $x\geq M$, contradicting the condition $\Psi_{\theta}\in L^2(\R^3,\langle x\rangle^{-1-\delta})$ for $\delta>0$. We deduce that $\theta\in\Upsilon(V)$, and Theorem \ref{main} provides the desired expression for $\alpha$. Moreover, given that $\psi_{\theta}(M)\neq 0$, we obtain $\alpha\neq\infty$.
\end{proof}

\section{Linear potentials}\label{linear}

In this Section, we focus on a family of compactly supported potentials, which are linear in a neighborhood of the origin. More precisely, for $\xi\in\R$, we consider
\begin{equation*}
 V_{\xi}(x) = (1 - \xi x) \mathbbm{1}_{[0,1]}.
\end{equation*}
Hence, given $\eps,\lambda\in\,^*\R$, with $\eps$ being a positive infinitesimal and $|\lambda|=O(\eps^{-2})$, we want study the Schr\"odinger operator
$$H^{(\xi)}_{\lambda,\eps}:=-\Delta_{D}+\lambda V_{\xi}\Big(\frac{x}{\eps}\Big),$$
which is self-adjoint on $^*\! L^2(\R^+)$ and near standard, as we proved in Section \ref{PMR}.

For $\theta\in\R\setminus\{0\}$, let $\psi^{(\xi)}_{\theta}\in\mathcal{C}^1(0,1)$ be the solution to the Cauchy problem 
\begin{equation}\label{problemLin}
\begin{cases}
-\partial_x^2\,\psi^{(\xi)}_{\theta}(x)+\theta V_{\xi}(x)\psi^{(\xi)}_{\theta}(x)=0,\quad x\in\,[0,1],\\[0.12cm]
\psi_{\theta}^{(\xi)}(0)=0,\\[0.12cm]
\partial_x\psi_{\theta}^{(\xi)}(0)=1.
\end{cases}
\end{equation}

In view of Theorem \ref{main}, we only need to consider the case when $\partial_x\psi_{\theta}^{(\xi)}(1)=0$, which defines the set $\Upsilon:=\Upsilon^{(\xi)}$. We distinguish between two cases, with respect to the value of the parameter $\xi$.

\textbf{Case I:} ($\xi=0$). In this case, we recover the square potential, which corresponds to the Schr\"{o}dinger operator
\begin{equation}
    H_\alpha = -\Delta_D + \lambda_\alpha \mathbbm{1}_{[0,\varepsilon]}
\end{equation}
analyzed by Albeverio, Fenstad, and H\o{}egh-Krohn in \cite{Albeverio-Fenstad-HoeghKrohn-1979_singPert_NonstAnal}. We have 
\begin{equation*}
\psi_{\theta}^{(0)}(x) =\frac{e^{\sqrt{\theta}x}-e^{-\sqrt{\theta}x}}{2\sqrt{\theta}}.
\end{equation*}

For a rectangular potential barrier ($\theta > 0$), it is straightforward to check that $\partial_x\psi_{\theta}^{(0)}(1)\neq 0$. When $\theta<0$ (corresponding to a rectangular potential well) we have
$$\partial_x\psi_{\theta}^{(0)}(1)=\cos(\sqrt{|\theta|}),$$
whence the set $\Upsilon^{(0)}$ is given explicitly by
$$\Upsilon^{(0)}=\Big\{-\pi^2\Big(k+\frac12\Big)^2\,|\,k\in\N\Big\}\subset(-\infty,0).$$
Due to Theorem~\ref{main}, $\lambda$ can be expressed as
$$\lambda=\frac{\theta}{\eps^2}+\frac{\omega}{\eps}+o(\eps^{-1}),\quad\theta\in\Upsilon^{(0)},\,\omega\in\R.$$
Then we find that  $\st(H^{(0)}_{\lambda,\eps})=-\Delta_{\alpha}$, with
 $$\alpha=\omega\big(\psi^{(0)}_{\theta}(M)\big)^{-2}\int_{0}^MV_0(t)\big(\psi^{(0)}_{\theta}\big)^2(t)dt=\frac{\omega}{2},$$
and we recover the result in \cite{Albeverio-Fenstad-HoeghKrohn-1979_singPert_NonstAnal}. As was noticed before, in this particular case, $\alpha$ does not depend on $\theta\in\Upsilon^{(0)}$.

\textbf{Case II:} ($\xi\neq 0$). The solution of~(\ref{problemLin}) reads
$$\psi^{(\xi)}_{\theta}(x) = \frac{1}{\xi\sigma_{\theta, \xi}}\frac{\operatorname{Bi}(\sigma_{\theta, \xi}) \operatorname{Ai}(\sigma_{\theta, \xi}(1 - \xi x)) - \operatorname{Ai}(\sigma_{\theta, \xi}) \operatorname{Bi}(\sigma_{\theta, \xi}(1 - \xi x))}{\operatorname{Bi}(\sigma_{\theta, \xi}) \operatorname{Ai}'(\sigma_{\theta, \xi}) - \operatorname{Ai}(\sigma_{\theta, \xi}) \operatorname{Bi}'(\sigma_{\theta, \xi})}, $$
where $\operatorname{Ai}(x)$ and $\operatorname{Bi}(x)$ are the Airy functions~\cite[Chapter 9.1]{Olver1997}, and $\sigma_{\theta, \xi} = \sqrt[3]{\frac{\theta}{\xi^2}}$. Taking into account the value of the Wronskian of the Airy functions, namely $W(\operatorname{Ai}(x),\operatorname{Bi}(x))=\frac{1}{\pi}$, we deduce
\begin{equation}\label{SolutionAiry}
\psi^{(\xi)}_{\theta}(x) = \frac{\pi}{\xi\sigma_{\theta, \xi}}\Bigl[\operatorname{Ai}(\sigma_{\theta, \xi}) \operatorname{Bi}(\sigma_{\theta, \xi}(1 - \xi x))- \operatorname{Bi}(\sigma_{\theta, \xi}) \operatorname{Ai}(\sigma_{\theta, \xi}(1 - \xi x)) \Bigr].
\end{equation}
Imposing the condition $\partial_x\psi^{(\xi)}_{\theta}(1)=0$, we obtain
$$\Upsilon^{(\xi)}=\Big\{\theta\,:\,\operatorname{Ai}(\sigma_{\theta, \xi}) \operatorname{Bi}'(\sigma_{\theta, \xi}(1 - \xi))- \operatorname{Bi}(\sigma_{\theta, \xi}) \operatorname{Ai}'(\sigma_{\theta, \xi}(1 - \xi))=0\Big\}.$$
It is easy to check that $\Upsilon$ is an infinite discrete set. Moreover, for $x>0$,
\begin{equation}\label{sign_Airy}
\frac{\operatorname{Ai}(x)}{\operatorname{Bi}(x)}>0,\qquad \frac{\operatorname{Ai}'(x)}{\operatorname{Bi}'(x)}<0.
\end{equation}
Using \eqref{sign_Airy} we deduce that, for $\xi\leq 1$, $\Upsilon^{(\xi)}\subset(-\infty,0)$. When $\xi>1$ the potential $V_{\xi}$ has a non-zero negative part, whence the set $\Upsilon^{(\xi)}$ contains also positive values.

Now, if $\lambda$ has the form
\begin{equation}
\lambda = \frac{\theta}{\eps^2}+\frac{\omega}{\eps}+o(\eps^{-1}),\quad\theta\in\Upsilon^{(\xi)},\,\omega\in\R,
\end{equation}
then $\st(H^{(\xi)}_{\lambda,\eps})=-\Delta_{\alpha^{(\xi)}}$, where 
 $$\alpha^{(\xi)} = \frac{\omega}{(\psi^{(\xi)}_{\theta}(1))^2}\int_{0}^1 V_{\xi}(x')\big(\psi^{(\xi)}_{\theta}(x')\big)^2 dx'.$$
The integral in the r.h.s. is equal to
$$ \int_{0}^1 V_{\xi}(x')\big(\psi^{(\xi)}_{\theta} (x') \big)^2dx' = \frac{1}{3 \xi^3 \sigma_{\theta, \xi}^3 } \left[ 1 - \sigma_{\theta, \xi} (1-\xi)^2 \Bigl(\frac{\operatorname{Ai}(\sigma_{\theta,\xi})}{\operatorname{Ai}'(\sigma_{\theta,\xi}(1-\xi))}\Bigr)^2\,\right],$$
and by computing the explicit value of $\psi^{(\xi)}_{\theta}(1)$ we obtain
\begin{equation}\label{alphaGeneral}
 \alpha^{(\xi)} = \frac{\omega}{3\xi \sigma_{\theta, \xi}} \Bigl[ \Bigl(\frac{\operatorname{Ai}'(\sigma_{\theta,\xi}(1-\xi))}{\operatorname{Ai}(\sigma_{\theta,\xi})}\Bigr)^2 - \sigma_{\theta, \xi} (1-\xi)^2 \Bigr].
\end{equation}
In particular, for the triangular potential (i.e. $\xi=1$) we have
\begin{equation}
 \alpha^{(1)} = \frac{\omega}{3\sqrt[3]{\theta}}\Bigl(\frac{\operatorname{Ai}'(0)}{\operatorname{Ai}(\sqrt[3]{\theta})}\Bigr)^2.
\end{equation}

We conclude this Section by focusing on what happens when potential $V_\xi$ is an infinitesimal modification of the characteristic function $\mathbbm{1}_{[0,1]}$. Namely, we consider the operator $H^{(\xi)}_{\lambda,\eps}$ in the case when $\xi \in ^*\!\mathbb{R}$ is infinitesimal, where we have $\st\left(V_{\xi}\right)=\mathbbm{1}_{[0,1]}$. Our aim is to show that, if $|\lambda|=O(\eps^{-2})$, then $H^{(\xi)}_{\lambda,\eps}$ is near standard, with $\st(H^{(\xi)}_{\lambda,\eps})=\st(H^{(0)}_{\lambda,\eps})$. In view of Definition \ref{de:ns}, it is enough to show that, for every $f\in L^2(\R^+)$ and $z\in\C\setminus\R$,
\begin{equation}\label{eq:eq:ris}
\st\big((H^{(\xi)}_{\lambda,\eps}-z)^{-1}\,\!^*\! f\big)\,|_{\R^+}=\st\big((H^{(0)}_{\lambda,\eps}-z)^{-1}\,\!^*\! f\big)\,|_{\R^+}.
\end{equation}
By means of the discussion in Section \ref{PMR}, this is equivalent to prove the following identities:
\begin{gather}
\label{id_no_der}\st\big(\psi^{(\xi)}_{\theta}(x)\big)=\psi^{(0)}_{\theta}(x),\quad\forall\,x\in[0,1],\,\forall\,\theta<0,\\
\label{id_der}\st\big(\partial_x\psi^{(\xi)}_{\theta}(1)\big)=\partial_x\psi^{(0)}_{\theta}(1),\quad\forall\,\theta\in\R\setminus\{0\}.
\end{gather}

Let us prove explicitly the identity \eqref{id_no_der}, the computations for \eqref{id_der} being similar. We use the following asymptotic expansions of the Airy functions, valid as $x\rightarrow +\infty$ (see e.g. \cite[Chapter 9.7]{Olver1997}).
\begin{eqnarray}
\label{Airyuno}\mathrm{Ai}\left(-x\right)&=\frac{1}{\sqrt{\pi}}\cos\big(\frac23 x^{3/2}-\frac{\pi}{4}\big)x^{-1/4}+O(x^{-7/4})&\\\label{Airydue}
\mathrm{Bi}\left(-x\right) &=-\frac{1}{\sqrt{\pi}}\sin\big(\frac23 x^{3/2}-\frac{\pi}{4}\big)x^{-1/4}+O(x^{-7/4})& 
\end{eqnarray}
Using \eqref{SolutionAiry}, \eqref{Airyuno}, \eqref{Airydue}, the identity $\sin(\alpha - \beta) = \sin(\alpha)\cos(\beta) - \cos(\alpha)\sin(\beta)$, and the Taylor expansion, we obtain
\begin{eqnarray}
\nonumber \operatorname{st}(\psi^{(\xi))}_{\theta}(x)) &=& \operatorname{st}\Biggl(\frac{\sin\Bigl(\frac{2}{3}(-\sigma_{\theta, \xi})^{3/2} (1 - (1-\xi x)^{3/2})\Bigr)}{\xi(-\sigma_{\theta, \xi})^{3/2} (1-\xi x)^{1/4}}\Biggr) \\
\nonumber &=& \operatorname{st} \Biggl(\frac{\sin\Bigl(\frac{2}{3}\sqrt{-\theta} \frac{1 - (1-\xi x)^{3/2}}{\xi}\Bigr)}{\sqrt{-\theta} (1-\xi x)^{1/4}} \Biggr)\\
\nonumber &=& \frac{\sin(\sqrt{-\theta} x )}{\sqrt{-\theta}} \equiv \psi^{(0)}_{\theta}(x).
\end{eqnarray}

\section{Conclusions}\label{Conclusions}
We have shown in this paper that non-Archimedean methods are well suited for the study of singular Schr\"odinger operators, analyzing in detail the case of delta-like interactions on the half-line. The transfer principle, indeed, provides a very efficient way to construct a self-adjoint operator with infinite magnitude and infinitesimal range, in the abstract framework of nonstandard analysis. Then, the finiteness of the resolvent allows to conclude that the nonstandard singular operator can be actually restricted, in a canonical way, to a classical Schr\"odinger operator with a Robin boundary condition at the origin. Translated into the standard setting, this provides an approximation result for the point interaction by means of a suitable family of re-scaled, regular potentials. Remarkably, the spectral conditions needed to produce a non-trivial boundary condition in the limit (identified in Theorem \ref{main} by the condition $\psi'_{\theta}(M)=0$ for the coupling parameter $\theta$) arise in this context by a quite direct argument, i.e.~by considering an eigenfunction expansion with respect to the infinitesimal parameter associated to the range of the interaction.

This approach appears to be quite versatile, and it can be adapted to the study of point interactions on bounded domains and compact manifolds, as well as to more complicated perturbations of the Laplace operator, such as interactions supported on curves and surfaces, non-local singular operators, measure-type potentials. In addition, it could be helpful in order to investigate problems arising from multi-particle quantum systems with contact interactions. In this context, in fact, there are many unsolved questions concerning the approximation of idealized zero-range models by means of regular Schr\"odinger operators, and the non-Archimedean point of view could provide useful insights and open new perspectives.

\section*{Conflict of interest}
On behalf of all authors, the corresponding author states that there is no conflict of interest.

\end{document}